\documentclass[letterpaper,UKenglish]{lipics}
\usepackage{microtype}%if unwanted, comment out or use option "draft"
\usepackage{amssymb,amsmath,latexsym,hyperref}%,lscape,}%,subfigure,multirow,hhline,}

\def\ML{\mathop{\rm ML}\nolimits}
\def\MA{\mathop{\rm MA}\nolimits}

\title{Single machine scheduling with job-dependent machine deterioration%
\footnote{This work was partially supported by AITF and NSERC Canada, and China Scholarship Council.}}
\titlerunning{Single machine scheduling with job-dependent machine deterioration} %optional, in case that the title is too long; the running title should fit into the top page column

\author[1,2]{Wenchang Luo}
\author[2]{Yao Xu}
\author[3]{Weitian Tong}
\author[2]{Guohui Lin}
\affil[1]{Faculty of Science, Ningbo University.
  Ningbo, Zhejiang 315211, China.}
\affil[2]{Department of Computing Science, University of Alberta.
  Edmonton, Alberta T6G 2E8, Canada.
  \texttt{\{wenchang,xu2,guohui\}@ualberta.ca}}
\affil[3]{Department of Computer Sciences, Georgia Southern University.
  Statesboro, Georgia 30460, USA.
  \texttt{wtong@georgiasouthern.edu}}
\authorrunning{Luo {\it et al.} version/\today} %mandatory. First: Use abbreviated first/middle names. Second (only in severe cases): Use first author plus 'et. al.'

\Copyright{Wenchang Luo, Yao Xu, Weitian Tong, and Guohui Lin}%mandatory, please use full first names. LIPIcs license is "CC-BY";  http://creativecommons.org/licenses/by/3.0/

\subjclass{Dummy classification -- please refer to \url{http://www.acm.org/about/class/ccs98-html}}% mandatory: Please choose ACM 1998 classifications from http://www.acm.org/about/class/ccs98-html . E.g., cite as "F.1.1 Models of Computation". 
% mandatory: Please provide 1-5 keywords
\begin{document}

\maketitle

\begin{abstract}
%===============================================================================
We consider the single machine scheduling problem with job-dependent machine deterioration.
In the problem, we are given a single machine with an initial non-negative maintenance level, 
and a set of jobs each with a non-preemptive processing time and a machine deterioration.
Such a machine deterioration quantifies the decrement in the machine maintenance level after processing the job.
To avoid machine breakdown, one should guarantee a non-negative maintenance level at any time point;
and whenever necessary, a maintenance activity must be allocated for restoring the machine maintenance level. 
The goal of the problem is to schedule the jobs and the maintenance activities such that the total completion time of jobs is minimized.  
There are two variants of maintenance activities: 
in the partial maintenance case each activity can be allocated to increase the machine maintenance level to any level not exceeding the maximum;
in the full maintenance case every activity must be allocated to increase the machine maintenance level to the maximum.
In a recent work, the problem in the full maintenance case has been proven NP-hard;
several special cases of the problem in the partial maintenance case were shown solvable in polynomial time,
but the complexity of the general problem is left open.
In this paper we first prove that the problem in the partial maintenance case is NP-hard, thus settling the open problem;
we then design a $2$-approximation algorithm.
\end{abstract}

\section{Introduction}\label{sec1}
%===============================================================================
In many scheduling problems, processing a job on a machine causes the machine to deteriorate to some extent,
and consequently maintenance activities need to be executed in order to restore the machine capacity.
Scheduling problems with maintenance activities have been extensively investigated since the work of Lee and Liman~\cite{LL92}.

A maintenance activity is normally described by two parameters, the starting time and the duration.
If these two parameters are given beforehand, a maintenance activity is referred to as {\em fixed};
otherwise it is called {\em flexible}.
Various scheduling models with fixed maintenance activities, on different machine environments and job characteristics,
have been comprehensively surveyed by Schmidt~\cite{Sch00}, Lee~\cite{Lee04}, and Ma {\it et al.}~\cite{MCZ10}.

A number of researchers initiated the work with flexible maintenance activities.
Qi {\it et al.}~\cite{QCF99} considered a single machine scheduling problem to simultaneously schedule jobs and maintenance activities,
with the objective to minimize the total completion time of jobs. 
They showed that the problem is NP-hard in the strong sense and proposed heuristics and a branch-and-bound exact algorithm.
(Qi \cite{Qi07} later analyzed the worst-case performance ratio for one of the heuristics, {\em the shortest processing time first} or SPT.)
Lee and Chen~\cite{LC00} studied the multiple parallel machines scheduling problem where each machine must be maintained exactly once,
with the objective to minimize the total weighted completion time of jobs.
They proved the NP-hardness for some special cases and proposed a branch-and-bound exact algorithm based on column generation;
the NP-hardness for the general problem is implied.
Kubzin and Strusevich~\cite{KS06} considered a two-machine open shop and a two-machine flow shop scheduling problems 
in which each machine has to be maintained exactly once and the duration of each maintenance depends on its starting time.
The objective is to minimize the maximum completion time of all jobs and all maintenance activities.
Among others, the authors showed that the open shop problem is polynomial time solvable for quite general functions 
defining the duration of maintenance in its starting time;
they also proved that the flow shop problem is binary NP-hard and presented a fully polynomial time approximation scheme (FPTAS)~\cite{KS06}.

Returning to a single machine scheduling problem,
Chen~\cite{Che08} studied the periodic maintenance activities of a constant duration not exceeding the available period,
with the objective to minimize the maximum completion time of jobs (that is, the {\em makespan}).
The author presented two mixed integer programs and heuristics and conducted computational experiments to examine their performance.
Mosheiov and Sarig~\cite{MS09} considered the problem where the machine needs to be maintained prior to a given deadline,
with the objective to minimize the total weighted completion time of jobs. 
They showed the binary NP-hardness and presented a pseudo-polynomial time dynamic programming algorithm and an efficient heuristic.
Luo {\it et al.}~\cite{LCZ10} investigated a similar variant (to \cite{MS09}) in which the jobs are weighted and
the duration of the maintenance is a nondecreasing function of the starting time (which must be prior to a given deadline).
Their objective is to minimize the total weighted completion time of jobs;
the authors showed the weakly NP-hardness,
and for the special case of concave duration function they proposed a $(1 + \sqrt{2}/2 + \epsilon)$-approximation algorithm.
Yang and Yang~\cite{YY10} considered a position-dependent aging effect described by a power function 
under maintenance activities and variable maintenance duration considerations simultaneously;
they examined two models with the objective to minimize the makespan, and for each of them they presented a polynomial time algorithm.

Scheduling on two identical parallel machines with periodic maintenance activities was examined by Sun and Li~\cite{SL10},
where the authors presented approximation algorithms with constant performance ratios for minimizing the makespan or minimizing the total completion time of jobs.
Xu {\it et al.}~\cite{XYL10} considered the case where the length of time between two consecutive maintenances is bounded;
they presented an approximation algorithm for the multiple parallel machines scheduling problem to minimize the completion time of the last maintenance,
and for the single machine scheduling problem to minimize the makespan, respectively.

\subsection{Problem definition}
%-------------------------------------------------------------------------------
Considering the machine deterioration in the real world, in a recent work by Bock {\it et al.}~\cite{BBH12},
a new scheduling model subject to {\em job-dependent machine deterioration} is introduced.
In this model, the single machine must have a non-negative {\em maintenance level} ($\ML$) at any time point, specifying its current maintenance state.
(A negative maintenance level indicates the machine breakdown, which is prohibited.) 
We are given a set of jobs ${\cal J} = \{J_i, i = 1, 2, \ldots, n\}$,
where each job $J_i = (p_i, \delta_i)$ is specified by its non-preemptive {\em processing time} $p_i$ and {\em machine deterioration} $\delta_i$. 
The machine deterioration $\delta_i$ quantifies the decrement in the machine maintenance level after processing the job $J_i$.
(That is, if before processing the job $J_i$ the maintenance level is $\ML$, then afterwards the maintenance level reduces to $\ML - \delta_i$ ---
suggesting that $\ML$ has to be at least $\delta_i$ in order for the machine to process the job $J_i$.)

Clearly, to process all the jobs, {\em maintenance activities} ($\MA$s) need to be allocated inside a schedule to restore the maintenance level,
preventing machine breakdown.
Given that the machine can have a maximum maintenance level of $\ML^{\max}$, and assuming a unit maintenance speed,
an $\MA$ of a duration $D$ would increase the maintenance level by $\min\{D, \ML^{\max} - \ML\}$, where $\ML$ is the maintenance level before the $\MA$.

With an initial machine maintenance level $\ML_0$, $0 \le \ML_0 \le \ML^{\max}$,
the goal of the problem is to schedule the jobs and necessary $\MA$s such that all jobs can be processed without machine breakdown,
and that the total completion time of jobs is minimized.

There are two variants of the problem depending on whether or not one has the freedom to choose the duration of an $\MA$:
in the {\em partial maintenance} case, the duration of each $\MA$ can be anywhere in between $0$ and $(\ML^{\max} - \ML)$,
where $\ML$ is the maintenance level before the $\MA$;
in the {\em full maintenance} case, however, the duration of every $\MA$ must be exactly $(\ML^{\max} - \ML)$,
consequently increasing the maintenance level to the maximum value $\ML^{\max}$.
Let $C_i$ denote the completion time of the job $J_i$, for $i = 1, 2, \ldots, n$.
In the three field notation, the two problems discussed in this paper are denoted as $(1 | p\MA | \sum_i C_i)$ and $(1 | f\MA | \sum_i C_i)$,
respectively, where $p\MA$ and $f\MA$ refer to the partial and the full maintenance, respectively.

\subsection{Prior work and our contribution}
%-------------------------------------------------------------------------------
Bock {\it et al.}~\cite{BBH12} proved that $(1 | f\MA | \sum_i C_i)$ is NP-hard, even when $p_i = p$ for all $i$ or when $p_i = \delta_i$ for all $i$,
both by a reduction from the {\sc Partition} problem~\cite{GJ79};
while all the jobs have the same deterioration, i.e. $\delta_i = \delta$ for all $i$, the problem can be solved in $O(n \log n)$ time.
For the partial maintenance case, Bock {\it et al.}~\cite{BBH12} showed that the SPT rule gives an optimal schedule for
$(1 | p\MA | \sum_i C_i)$ when $p_i < p_j$ implies $p_i + \delta_i \le p_j + \delta_j$ for each pair of $i$ and $j$
(which includes the special cases where $p_i = p$ for all $i$, or $\delta_i = \delta$ for all $i$, or $p_i = \delta_i$ for all $i$).
The complexity of the general problem $(1 | p\MA | \sum_i C_i)$ was left as an open problem.
Also, to the best of our knowledge, no approximation algorithms have been designed for either problem.

Our main contribution in this paper is to settle the NP-hardness of the general problem $(1 | p\MA | \sum_i C_i)$.
Such an NP-hardness might appear a bit surprising at the first glance since one has so much freedom in choosing the starting time and the duration of each $\MA$.
Our reduction is from the {\sc Partition} problem too, using a kind of job swapping argument.
This reduction is presented in Section 3, following some preliminary properties we observe for the problem in Section 2.
In Section 4, we propose a $2$-approximation algorithm for $(1 | p\MA | \sum_i C_i)$.
We conclude the paper in Section 5 with some discussion on the (in-)approximability.

Lastly, we would like to point out that when the objective is to minimize the makespan $C_{\max}$, i.e. the maximum completion time of jobs,
$(1 | p\MA | C_{\max})$ can be trivially solved in $O(n)$ time and
$(1 | f\MA | C_{\max})$ is NP-hard but admits an $O\left(n^2 (\ML^{\max})^2 \log \left(\sum_{i=1}^n (p_i + \delta_i)\right)\right)$ time algorithm
based on dynamic programming (and thus admits an FPTAS)~\cite{BBH12}.

\section{Preliminaries}
%===============================================================================
Given a feasible schedule $\pi$ to the problem $(1 | p\MA | \sum_i C_i)$,
which specifies the start processing time for each job and the starting time and the duration of each $\MA$,
we abuse slightly $\pi$ to also denote the permutation of the job indices $(1, 2, \ldots, n)$ in which the jobs are processed in order:
$\pi = (\pi_1, \pi_2, \ldots, \pi_n)$.
The following lemma is proved in \cite{BBH12}.

\begin{lemma}
\label{lemma01}
{\rm \cite{BBH12}}
There is an optimal schedule $\pi$ to $(1 | p\MA | \sum_i C_i)$ such that
the total maintenance duration before processing the job $J_{\pi_i}$ equals $\max\left\{0, \sum_{j = 1}^i \delta_{\pi_j} - \ML_0\right\}$,
for each $i = 1, 2, \ldots, n$.
\end{lemma}

Lemma~\ref{lemma01} essentially states that each $\MA$ should be pushed later in the schedule as much as possible until absolutely necessary,
and its duration should be minimized just for processing the succeeding job.
In the sequel, we limit our discussion on the feasible schedules satisfying these two properties.
We define the {\em separation job} in such a schedule $\pi$ as the first job that requires an $\MA$ (of a positive duration).

\begin{lemma}
\label{lemma02}
Suppose $J_{\pi_k}$ is the separation job in an optimal schedule $\pi$ to $(1 | p\MA | \sum_i C_i)$.
Then,
\begin{itemize}
\parskip=0pt
\item
	the jobs before the separation job $J_{\pi_k}$ are scheduled in the SPT order;
\item
	the jobs after the separation job $J_{\pi_k}$ are scheduled in the {\em shortest sum-of-processing-time-and-deterioration first} (SSF) order;
\item
	the jobs adjacent to the separation job $J_{\pi_k}$ satisfy
	\[
	 p_{\pi_{k-1}} + \min\{\delta_{\pi_{k-1}}, \delta_{\pi_k} - \delta\} \le
	 p_{\pi_k} + (\delta_{\pi_k} - \delta) \le
	 p_{\pi_{k+1}} + \max\{0, \delta_{\pi_{k+1}} - \delta\},
	\]
	where $\delta = \ML_0 - \sum_{i=1}^{k-1} \delta_{\pi_i}$ is the remaining maintenance level before the first $\MA$.
\end{itemize}
\end{lemma}
\begin{proof}
Starting with an optimal schedule satisfying the properties stated in Lemma~\ref{lemma01},
one may apply a simple job swapping procedure if the job order is violated
either in the prefix or in the suffix of job order separated by the separation job $J_{\pi_k}$.
This procedure would decrease the value of the objective, contradicting to the optimality.
That is, we have (see Figure~\ref{fig01} for an illustration)
\begin{eqnarray}
\label{eq1}
p_{\pi_1} \le p_{\pi_2} \le \ldots \le p_{\pi_{k-1}}, \ \mbox{ and}\\
\label{eq2}
p_{\pi_{k+1}} + \delta_{\pi_{k+1}} \le p_{\pi_{k+2}} + \delta_{\pi_{k+2}} \le \ldots \le p_{\pi_n} + \delta_{\pi_n}.
\end{eqnarray}

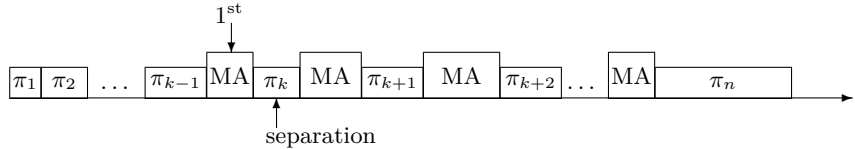
\begin{figure}[htb]
\begin{center}
\unitlength=0.4mm
\begin{picture}(300, 50)
\put(0, 10){\vector(1, 0){280}}
\put(0, 10){\framebox(10, 10){\small $\pi_1$}}
\put(10.5, 10){\framebox(15, 10){\small $\pi_2$}}
\put(30, 13){\small $\ldots$}
\put(45, 10){\framebox(20, 10){\small $\pi_{k-1}$}}
\put(65.5, 10){\framebox(15, 15){\small $\MA$}}
\put(81, 10){\framebox(15, 10){\small $\pi_k$}}
\put(88.5, 0){\vector(0, 1){10}}
\put(85, -5){\small separation}
\put(73.5, 35){\vector(0, -1){10}}
\put(68, 35){\small 1$^{\rm st}$}
\put(96.5, 10){\framebox(20, 15){\small $\MA$}}
\put(117, 10){\framebox(20, 10){\small $\pi_{k+1}$}}
\put(137.5, 10){\framebox(25, 15){\small $\MA$}}
\put(163, 10){\framebox(20, 10){\small $\pi_{k+2}$}}
\put(185, 13){\small $\ldots$}
\put(199, 10){\framebox(15, 15){\small $\MA$}}
\put(214.5, 10){\framebox(45, 10){\small $\pi_{n}$}}
\end{picture}
\end{center}
\caption{An illustration of the optimal schedule $\pi$ stated in Lemma~\ref{lemma02}, where the separation job is $J_{\pi_k}$;
	the width of a framebox does not necessarily equal the processing time of a job or the duration of an $\MA$.\label{fig01}}
\end{figure}

Let $\delta =  \ML_0 - \sum_{i=1}^{k-1} \delta_{\pi_i}$ denote the remaining maintenance level before the first $\MA$.
Because $\delta < \delta_{\pi_k}$, an (the first) $\MA$ of duration $\delta_{\pi_k} - \delta$ needs to be performed for processing the separation job $J_{\pi_k}$.
From the optimality of $\pi$, swapping the two jobs $J_{\pi_k}$ and $J_{\pi_{k+1}}$ should not decrease the objective, that is,
\[
\left\{
\begin{array}{ll}
p_{\pi_k} + (\delta_{\pi_k} - \delta) \le p_{\pi_{k+1}} + (\delta_{\pi_{k+1}} - \delta), &\mbox{if } \delta_{\pi_{k+1}} > \delta;\\
p_{\pi_k} + (\delta_{\pi_k} - \delta) \le p_{\pi_{k+1}}, &\mbox{otherwise}.
\end{array}\right.
\]
Similarly, swapping the two jobs $J_{\pi_{k-1}}$ and $J_{\pi_k}$ should not decrease the objective, that is,
\[
\left\{
\begin{array}{ll}
p_{\pi_{k-1}} \le p_{\pi_k}, &\mbox{if } \delta_{\pi_{k-1}} \ge \delta_{\pi_k} - \delta;\\
p_{\pi_{k-1}} + \delta_{\pi_{k-1}} \le p_{\pi_k} + (\delta_{\pi_k} - \delta), &\mbox{otherwise}.
\end{array}\right.
\]
These together give
\begin{equation}
\label{eq3}
p_{\pi_{k-1}} + \min\{\delta_{\pi_{k-1}}, \delta_{\pi_k} - \delta\} \le
 p_{\pi_k} + (\delta_{\pi_k} - \delta) \le
 p_{\pi_{k+1}} + \max\{0, \delta_{\pi_{k+1}} - \delta\}.
\end{equation}
This proves the lemma.
\end{proof}

From Lemma~\ref{lemma02}, one sees that the separation job in an optimal schedule is unique,
in the sense that it cannot always be ``appended'' to either the prefix SPT order or the suffix SSF order.
This is reflected in our NP-completeness reduction in Section 3, where we force a certain scenario to happen.

\section{NP-hardness of the problem $(1 | p\MA | \sum_i C_i)$}\label{sec2}
%===============================================================================
Our reduction is from the classic NP-complete problem {\sc Partition}~\cite{GJ79}, formally defined as follows:
\begin{quote}
{\sc Partition:}

{\sc Instance:} A set $X$ of $n$ positive integers $X = \{x_1, x_2, \ldots, x_n\}$, with $\sum_{i=1}^n x_i = 2B$.

{\sc Query:} Is there a subset $X_1 \subset X$ such that $\sum_{x \in X_1} x = \sum_{x \in X - X_1} x = B$?
\end{quote}

We abuse $X$ to denote the instance of {\sc Partition} with the set $X = \{x_1, x_2, \ldots, x_n\}$ and $\sum_{i=1}^n x_i = 2B$.
The corresponding instance $I$ of the problem $(1 | p\MA | \sum_i C_i)$ is constructed in polynomial time, as follows:
\begin{center}
\begin{tabular*}{\textwidth}{rl}
Number of jobs:				& $2n + 3$;\\
Job processing time:		& $p_{n+1+i} = p_i = \sum_{j=1}^i x_j, \mbox{ for } i = 0, 1, 2, \ldots, n$,\\
							& $p_{2n+2} = M - 2B$;\\
Machine deterioration:		& $\delta_{n+1+i} = \delta_i = M - 2p_i, \mbox{ for } i = 0, 1, 2, \ldots, n$,\\
							& $\delta_{2n+2} = 0$;\\
Initial maintenance level:	& $\ML_0 = \sum_{i=0}^n \delta_i - 2B$;\\
Maximum maintenance level:	& $\ML^{\max} = \sum_{i=0}^n \delta_i$;\\
Objective threshold:		& $Q = Q_0 + B$,
\end{tabular*}
\end{center}
(note that $p_{n+1} = p_0 = \sum_{j=1}^0 x_j = 0$ due to the empty range for $j$)
where $M$ is a big integer:
\begin{equation}
\label{eq4}
M > (4n + 8) B,
\end{equation}
and $Q_0$ is the total completion time of jobs for an {\em initial infeasible schedule} $\pi^0$ (see Figure~\ref{fig02}):
\begin{equation}
\label{eq5}
Q_0 = \sum_{j=0}^n (n-j+1) p_j + (n+2)\left(\sum_{j=0}^n p_j + 2B + p_{2n+2}\right) + \sum_{j=0}^n (j+1)(p_{n+1+j} + \delta_{n+1+j}).
\end{equation}

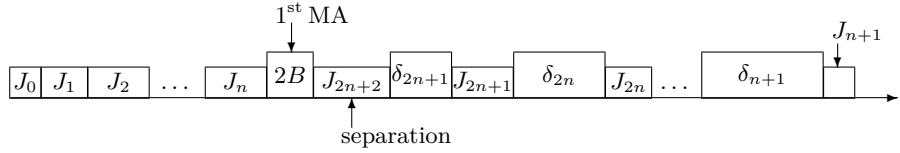
\begin{figure}[htb]
\begin{center}
\unitlength=0.4mm
\begin{picture}(300, 50)
\put(0, 10){\vector(1, 0){295}}
\put(0, 10){\framebox(10, 10){\small $J_0$}}
\put(10.5, 10){\framebox(15, 10){\small $J_1$}}
\put(26, 10){\framebox(20, 10){\small $J_2$}}
\put(50, 13){\small $\ldots$}
\put(65, 10){\framebox(20, 10){\small $J_n$}}
\put(85.5, 10){\framebox(15, 15){\small $2B$}}
\put(101, 10){\framebox(25, 10){\small $J_{2n+2}$}}
\put(113.5, 0){\vector(0, 1){10}}
\put(110, -5){\small separation}
\put(93.5, 35){\vector(0, -1){10}}
\put(88, 35){\small 1$^{\rm st} \MA$}
\put(126.5, 10){\framebox(20, 15){\small $\delta_{2n+1}$}}
\put(147, 10){\framebox(20, 10){\small $J_{2n+1}$}}
\put(167.5, 10){\framebox(30, 15){\small $\delta_{2n}$}}
\put(198, 10){\framebox(15, 10){\small $J_{2n}$}}
\put(216, 13){\small $\ldots$}
\put(230, 10){\framebox(40, 15){\small $\delta_{n+1}$}}
\put(270.5, 10){\framebox(10, 10){\small ~}}
\put(275, 28){\vector(0, -1){8}}
\put(273, 30){\small $J_{n+1}$}
\end{picture}
\end{center}
\caption{The initial infeasible schedule $\pi^0$ for the instance $I$ with the separation job $J_{2n+2}$;
	$\pi^0$ satisfies all properties stated in Lemma~\ref{lemma02}.
	All $\MA$s are indicated by their respective durations (for the first $\MA$, its duration is $\delta_{2n+2} - \delta = 2B$).\label{fig02}}
\end{figure}

The job order in this initial schedule $\pi^0$ is $(J_0, J_1, \ldots, J_n, J_{2n+2}, J_{2n+1}, J_{2n}, \ldots, J_{n+1})$,
and the first $\MA$ precedes the job $J_{2n+2}$, which is regarded as the separation job (see Figure~\ref{fig02}).
Before the separation job $J_{2n+2}$, the machine maintenance level is allowed to go into negative, but has to be restored to zero just for processing $J_{2n+2}$;
afterwards, machine breakdown is no longer tolerated.
From $\ML_0 = \sum_{i=0}^n \delta_i - 2B$, we know that $\pi^0$ is {\em infeasible} due to machine breakdown before the first $\MA$;
we will convert it to a feasible schedule later.
The {\sc Query} of the decision version of the problem $(1 | p\MA | \sum_i C_i)$ is
whether or not there exists a feasible schedule $\pi$ such that the total completion time of jobs is no more than $Q = Q_0 + B$.

Despite the infeasibility, the initial schedule $\pi^0$ has all the properties stated in Lemma~\ref{lemma02},
with the separation job $J_{2n+2}$ at the {\em center position}.
The first $(n+1)$ jobs are in the SPT order and the last $(n+1)$ jobs are in the SSF order;
since $\delta = -2B$, $p_n = p_{2n+1} = 2B$, $\delta_n = \delta_{2n+1} = M - 4B$, $p_{2n+2} = M - 2B$, $\delta_{2n+2} = 0$,
Eq.~(\ref{eq3}) is also satisfied due to the big $M$ in Eq.~(\ref{eq4}):
\[
p_n + \min\{\delta_n, \delta_{2n+2} - \delta\} < p_{2n+2} + (\delta_{2n+2} - \delta) = p_{2n+1} + \max\{0, \delta_{2n+1} - \delta\}.
\]
In the rest of the section, we will show that there is a subset $X_1 \subset X$ of sum exactly $B$ if and only if
the initial schedule $\pi^0$ can be converted into a feasible schedule $\pi$ with the total completion time of jobs no more than $Q = Q_0 + B$,
through a {\em repeated job swapping} procedure.

Notice that the two jobs $J_i$ and $J_{n+1+i}$ are identical, for $i = 0, 1, \ldots, n$.
In any schedule with the job $J_{2n+2}$ at the center position, if exactly one of $J_i$ and $J_{n+1+i}$ is scheduled before $J_{2n+2}$,
then we always say $J_i$ is scheduled before $J_{2n+2}$ while $J_{n+1+i}$ is scheduled after $J_{2n+2}$.
Also, when the two jobs $J_i$ and $J_{n+1 + i}$ are both scheduled before $J_{2n+2}$, then $J_{n+1 + i}$ precedes $J_i$;
when the two jobs $J_i$ and $J_{n+1 + i}$ are both scheduled after $J_{2n+2}$, then $J_i$ precedes $J_{n+1 + i}$.

\subsection{Proof of ``only if''}
%-------------------------------------------------------------------------------
In this subsection, we show that if there is a subset $X_1 \subset X$ of sum exactly $B$,
then the initial infeasible schedule $\pi^0$ can be converted into a feasible schedule $\pi$ with the total completion time no more than $Q = Q_0 + B$.
We also demonstrate the repeated job swapping procedure leading to this successful schedule $\pi$.

Suppose the indices of the elements in the subset $X_1$ are $\{i_1, i_2, \ldots, i_m\}$, satisfying $1 \le i_1 < i_2 < \ldots < i_m \le n$.
Starting with the initial schedule $\pi^0$, we sequentially swap the job $J_{i_\ell - 1}$ with the job $J_{n+1 + i_\ell}$, for $\ell = 1, 2, \ldots, m$.
Let $\pi^{\ell}$ denote the schedule after the $\ell$-th job swapping.

\begin{lemma}
\label{lemma03}
For each $1 \le \ell \le m$,
\begin{itemize}
\parskip=0pt
\item
	the schedule $\pi^{\ell}$ with the separation job $J_{2n+2}$ satisfies the properties in Lemma~\ref{lemma02};
\item
	the $\ell$-th job swapping decreases the total machine deterioration before the separation job $J_{2n+2}$ by $2 x_{i_\ell}$;
\item
	the $\ell$-th job swapping increases the total completion time by $x_{i_\ell}$.
\end{itemize}
\end{lemma}
\begin{proof}
Recall that the two jobs $J_{i_\ell}$ and $J_{n+1 + i_\ell}$ are identical.
Before the $\ell$-th job swapping between $J_{i_\ell - 1}$ and $J_{n+1 + i_\ell}$ (in the schedule $\pi^{\ell-1}$),
the jobs in between $J_{i_\ell - 1}$ and $J_{n+1 + i_\ell}$ are
\[
(J_{i_\ell - 1}, J_{i_\ell}, J_{i_\ell + 1}, \ldots, J_n, J_{2n+2}, J_{2n+1}, J_{2n}, \ldots, J_{n+1 + i_\ell + 1}, J_{n+1 + i_\ell}).
\]
After the swapping (in the schedule $\pi^\ell$) this sub-schedule becomes
\[
(J_{n+1 + i_\ell}, J_{i_\ell}, J_{i_\ell + 1}, \ldots, J_n, J_{2n+2}, J_{2n+1}, J_{2n}, \ldots, J_{n+1 + i_\ell + 1}, J_{i_\ell - 1}).
\]
By a simple induction, all jobs before $J_{n+1 + i_\ell}$ have their processing times less than $p_{i_\ell}$,
and thus the jobs before the separation job $J_{2n+2}$ are in the SPT order;
for a similar reason, the jobs after the separation job $J_{2n+2}$ are in the SSF order.

By the $\ell$-th job swapping, the change in the total machine deterioration before the separation job $J_{2n+2}$ is
$\delta_{i_\ell} - \delta_{i_\ell - 1} = - 2(p_{i_\ell} - p_{i_\ell - 1}) = - 2x_{i_\ell}$, that is, decreases by $2x_{i_\ell}$.
Therefore the duration of the first $\MA$ also decreases by $2 x_{i_\ell}$.
Since $J_n$ always directly precedes $J_{2n+2}$ and $p_n < p_{2n+2}$, the first half of Eq.~(\ref{eq3}) holds;
since $p_{2n+2} + \delta_{2n+2}$ is the smallest among all jobs, the second half of Eq.~(\ref{eq3}) holds.
That is, the schedule $\pi^\ell$ satisfies all properties in Lemma~\ref{lemma02}.

For ease of presentation, let $C_i$ denote the completion time of the job $J_i$ in the schedule $\pi^\ell$, and
let $C'_i$ denote the completion time of the job $J_i$ in the schedule $\pi^{\ell - 1}$.
Comparing to the schedule $\pi^{\ell - 1}$ ($\ell \ge 1$),
after the $\ell$-th job swapping between $J_{i_\ell - 1}$ and $J_{n+1 + i_\ell}$,
\begin{itemize}
\parskip=0pt
\item
	the completion time of jobs preceding $J_{n+1 + i_\ell}$ is unchanged;
\item
	$C_{n+1 + i_\ell} - C'_{i_\ell - 1} = p_{i_\ell} - p_{i_\ell - 1} = x_{i_\ell}$;
\item
	the completion time of each job in between $J_{i_\ell}$ and $J_n$ (inclusive, $n - i_\ell +1$ of them) increases by $x_{i_\ell}$;
\item
	the duration of the first $\MA$ decreases by $2 x_{i_\ell}$;
\item
	the completion time of each job in between $J_{2n+2}$ and $J_{n+1+i_\ell + 1}$ (inclusive, $n - i_\ell +1$ of them) decreases by $x_{i_\ell}$;
\item
	$C_{i_\ell - 1} - C'_{n+1+i_\ell} = - x_{i_\ell} + (\delta_{i_\ell - 1} + p_{i_\ell - 1}) - (\delta_{i_\ell} + p_{i_\ell}) = 0$;
\item
	from the last item, the completion time of jobs succeeding $J_{i_\ell - 1}$ is unchanged.
\end{itemize}
In summary, there are $(n - i_\ell + 2)$ jobs of which the completion time increases by $x_{i_\ell}$ and
$(n - i_\ell + 1)$ jobs of which the completion time decreases by $x_{i_\ell}$.
Therefore, the $\ell$-th job swapping between $J_{i_\ell - 1}$ and $J_{n+1 + i_\ell}$ increases the total completion time by $x_{i_\ell}$.
This finishes the proof.
\end{proof}

\begin{theorem}
\label{thm04}
If there is a subset $X_1 \subset X$ of sum exactly $B$,
then there is a feasible schedule $\pi$ to the instance $I$ with the total completion time no more than $Q = Q_0 + B$.
\end{theorem}
\begin{proof}
Let the indices of the elements in the subset $X_1$ be $\{i_1, i_2, \ldots, i_m\}$, such that $1 \le i_1 < i_2 < \ldots < i_m \le n$.
Starting with the initial schedule $\pi^0$, we sequentially swap the job $J_{i_\ell - 1}$ with the job $J_{n+1 + i_\ell}$, for $\ell = 1, 2, \ldots, m$.
Let $\pi^{\ell}$ denote the schedule after the $\ell$-th job swapping, and let $Q_\ell$ denote the total completion time of jobs in $\pi^\ell$.

From Lemma~\ref{lemma03} we know that the ending schedule $\pi^m$ satisfies all the properties in Lemma~\ref{lemma02}.
Also, the total machine deterioration before the separation job $J_{2n+2}$ in $\pi^m$ is
\[
\sum_{i=0}^n \delta_i - 2 \sum_{\ell=1}^m x_{i_\ell} = \sum_{i=0}^n \delta_i - 2B = \ML_0,
\]
suggesting that $\pi^m$ is a feasible schedule.
(The first $\MA$ has zero duration and thus becomes unnecessary.)

Moreover, the total completion time of jobs in $\pi^m$ is $Q_m = Q_0 + \sum_{\ell=1}^m x_{i_\ell} = Q_0 + B$.
Therefore, the schedule $\pi^m$ obtained from the initial schedule $\pi^0$ through the repeated job swapping procedure is a desired one.
\end{proof}

\subsection{Proof of ``if''}
%-------------------------------------------------------------------------------
In this subsection, we show that if there is a feasible schedule $\pi$ to the constructed instance $I$ with the total completion time no more than $Q = Q_0 + B$,
then there is a subset $X_1 \subset X$ of sum exactly $B$.
Assume without loss of generality that the schedule $\pi$ satisfies the properties in Lemma~\ref{lemma02}.
We start with some structure properties which the schedule $\pi$ must have.

\begin{lemma}
\label{lemma05}
Excluding the job $J_{2n+2}$,
there are at least $n$ and at most $(n+1)$ jobs scheduled before the first $\MA$ in the schedule $\pi$.
\end{lemma}
\begin{proof}
Recall that in Eq.~(\ref{eq4}) we set $M$ to be a large value such that $M > (4n + 8)B$.
Using $M > (4n + 6)B$, it follows from $M - 4B = \delta_n < \delta_{n-1} < \ldots < \delta_1 < \delta_0 = M$ that the initial machine maintenance level
\[
\ML_0 = \sum_{i=0}^n \delta_i - 2B > (n+1) (M - 4B) - 2B = n M + M - (4n + 6)B > n M.
\]
We thus conclude that at least $n$ jobs, excluding $J_{2n+2}$ which has $0$ deterioration, can be processed before the first $\MA$.

Nevertheless, if there were more than $(n+1)$ jobs scheduled before the first $\MA$, excluding $J_{2n+2}$, then their total machine deterioration would be
greater than $(n + 2)(M - 4B)$.
Using $M > (4n + 8)B$, we have
\[
(n + 2)(M - 4B) = (n + 1)M + M - (4n + 8)B > (n + 1)M > \sum_{i=0}^n \delta_i > \ML_0,
\]
contradicting the feasibility of the schedule $\pi$.
\end{proof}

\begin{lemma}
\label{lemma06}
There are at most $(n+1)$ jobs scheduled after the job $J_{2n+2}$ in the schedule $\pi$.
\end{lemma}
\begin{proof}
We prove the lemma by contradiction.
Firstly, noting that the job $J_{2n+2}$ has a much larger processing time compared to any other job ($M - 2B$ versus $2B$),
we conclude that the earliest possible position for $J_{2n+2}$ in the schedule $\pi$ is right before the first $\MA$.
We disallow a zero-duration $\MA$ and thus the job $J_{2n+2}$ can never be the separation job in $\pi$ due to $\delta_{2n+2} = 0$.

If $J_{2n+2}$ is scheduled after the separation job, by Eq.~(\ref{eq2}) or the SSF rule, for every job $J_i$ scheduled after $J_{2n+2}$ we have
$p_{2n+2} \le p_i + \delta_i$.
If $J_{2n+2}$ is scheduled right before the first $\MA$, by Eq.~(\ref{eq3}), for the separation job $J_i$ we have
$p_{2n+2} \le p_i + (\delta_i - \delta)$;
by Eqs.~(\ref{eq2}) and (\ref{eq3}), for every other job $J_i$ scheduled after $J_{2n+2}$ we have
$p_{2n+2} \le p_i + \delta_i$.
Therefore, the completion time of a job scheduled $\ell$ positions after the job $J_{2n+2}$ is at least $(\ell + 1) \times p_{2n+2}$.
If there were $(n+2)$ jobs scheduled after $J_{2n+2}$, then the total completion time of the last $(n+3)$ jobs would be at least
\[
\sum_{\ell=0}^{n+2} (\ell+1) p_{2n+2} = \frac {(n+3)(n+4)}2 p_{2n+2} = \frac {(n+3)(n+4)}2 (M - 2B).
\]

However, using $p_j \le 2B$ for $j \ne 2n+2$, one sees that Eq.~(\ref{eq5}) can be simplified as
\[
\begin{array}{rcl}
Q_0	&=	&\displaystyle\sum_{j=0}^n (n-j+1) p_j + (n+2)\left(\sum_{j=0}^n p_j + 2B + (M - 2B)\right) + \sum_{j=0}^n (j+1)(M - p_j)\\
	&=	&\displaystyle\frac {(n+2)(n+3)}2 M + 2\sum_{j=0}^n (n-j+1) p_j\\
	&\le&\displaystyle\frac {(n+2)(n+3)}2 M + 4B\sum_{j=0}^n (n-j+1)\\
	&=	&\displaystyle\frac {(n+2)(n+3)}2 M + 2(n+1)(n+2)B.
\end{array}
\]
Using $M > (3n+6)B$,
\[
\frac {(n+3)(n+4)}2 (M - 2B) \ge Q_0 + (n+3)M - (3n^2 + 13n + 16) B > Q_0 + (2n+2)B,
\]
that is, we would have the total completion time of the last $(n+3)$ jobs in $\pi$ strictly greater than $Q = Q_0 + B$,
contradicting to our assumption.
\end{proof}

Combining Lemmas~\ref{lemma05} and \ref{lemma06}, we have the following lemma regarding the position of $J_{2n+2}$ in the schedule $\pi$.

\begin{lemma}
\label{lemma07}
In the schedule $\pi$, the position of the job $J_{2n+2}$ has three possibilities:
\begin{description}
\parskip=0pt
\item[Case 1:]
	There are $(n+1)$ jobs before the first $\MA$, $\pi_{n+2} = 2n+2$, and $J_{\pi_{n+3}}$ is the separation job.
\item[Case 2:]
	There are $(n+1)$ jobs before the first $\MA$, $J_{\pi_{n+2}}$ is the separation job, and $\pi_{n+3} = 2n+2$.
\item[Case 3:]
	There are $n$ jobs before the first $\MA$, $J_{\pi_{n+1}}$ is the separation job, and $\pi_{n+2} = 2n+2$.
\end{description}
\end{lemma}
\begin{proof}
Note that the processing time of the job $J_{2n+2}$ is strictly greater than that of any other job,
while the sum of its processing time and machine deterioration ($p_{2n+2} + \delta_{2n+2}$) achieves the minimum.
Because $J_{2n+2}$ cannot act as the separation job due to $\delta_{2n+2} = 0$,
by Lemma~\ref{lemma02} it can only be either the last job scheduled before the first $\MA$ or the first job scheduled after the separation job
(through a possible job swapping, if necessary).
Using Lemmas~\ref{lemma05} and \ref{lemma06}, it is easy to distinguish the three possible cases stated in the lemma.
\end{proof}

Recall that the job order in the initial infeasible schedule $\pi^0$ is $(J_0, J_1, \ldots, J_n, J_{2n+2}, J_{2n+1},$ $J_{2n}, \ldots, J_{n+2}, J_{n+1})$,
and the first $\MA$ is executed before processing the job $J_{2n+2}$, which is regarded as the separation job (see Figure~\ref{fig02}).
In the sequel, we will again convert $\pi^0$ into our target schedule $\pi$ through a repeated job swapping procedure.
During such a procedure, the job $J_{2n+2}$ is kept at the center position,
and a job swapping always involves a job before $J_{2n+2}$ and a job after $J_{2n+2}$.

In Cases 1 and 3 of the schedule $\pi$, the job $J_{2n+2}$ is at the center position (recall that there are in total $2n + 3$ jobs),
and therefore the target schedule is well set.
In Case 2, $J_{2n+2}$ is at position $n+3$, not the center position;
we first exchange $J_{2n+2}$ and $J_{\pi_{n+2}}$ to obtain a schedule $\pi'$, which becomes our target schedule.
That is, we will first convert $\pi^0$ into $\pi'$ through a repeated job swapping procedure, 
and at the end exchange $J_{2n+2}$ back to the position $n+3$ to obtain the final schedule $\pi$.
In summary, our primary goal is to convert the schedule $\pi^0$ through a repeated job swapping procedure,
keeping the job $J_{2n+2}$ at the center position and keeping the first $\MA$ right before the job $J_{2n+2}$ (to be detailed next).
At the end, to obtain the target schedule $\pi$,
in Case 1, we swap the job $J_{2n+2}$ and the first $\MA$ ({\it i.e.}, moving the first $\MA$ one position backward);
in Case 2, we swap $J_{2n+2}$ and the immediate succeeding $\MA$ and the following job (with the $\MA$ merged with the first $\MA$);
in Case 3, we swap the first $\MA$ and its immediate preceding job ({\it i.e.}, moving the first $\MA$ one position forward).

In the target schedule ($\pi$ in Cases 1 and 3, or $\pi'$ in Case 2),
let $R = \{r_1, r_2, \ldots, r_m\}$ denote the subset of indices such that both $J_{r_j}$ and $J_{n+1 + r_j}$ are among the first $(n+1)$ jobs,
where $0 \le r_1 < r_2 < \ldots < r_m \le n$,
and $L = \{\ell_1, \ell_2, \ldots, \ell_m\}$ denote the subset of indices such that both $J_{\ell_j}$ and $J_{n+1 + \ell_j}$ are among the last $(n+1)$ jobs,
where $0 \le \ell_1 < \ell_2 < \ldots < \ell_m \le n$.
Note that $J_{2n+2}$ is at the center position in the target schedule, and thus it has to be $|R| = |L|$ and we let $m = |R|$.
Clearly, all these $\ell_j$'s and $r_j$'s are distinct from each other.

In the repeated job swapping procedure leading the initial infeasible schedule $\pi^0$ to the target feasible schedule,
the $j$-th job swapping is to swap the two jobs $J_{\ell_j}$ and $J_{n+1 + r_j}$.
The resultant schedule after the $j$-th job swapping is denoted as $\pi^j$, for $j = 1, 2, \ldots, m$.
In Section 3.1, the job swapping is ``{\em regular}'' in the sense that $\ell_j = r_j - 1$ for all $j$,
but now $\ell_j$ and $r_j$ do not necessarily relate to each other.
We remark that immediately after the swapping, a job sorting is needed to restore the SPT order for the prefix and the SSF order for the suffix
(see the last paragraph before Section 3.1 for possible re-indexing the jobs).

The following Lemma~\ref{lemma08} on the $j$-th job swapping, when $\ell_j < r_j$, is an extension of Lemma~\ref{lemma03}.

\begin{lemma}
\label{lemma08}
For each $1 \le j \le m$, if the schedule $\pi^{j-1}$ satisfies the first two properties in Lemma~\ref{lemma02} and $\ell_j < r_j$, then
\begin{itemize}
\parskip=0pt
\item
	the schedule $\pi^j$ satisfies the first two properties in Lemma~\ref{lemma02};
\item
	the $j$-th job swapping decreases the total machine deterioration before the center job $J_{2n+2}$ by
	$\delta_{\ell_j} - \delta_{r_j} = 2 \sum_{k=\ell_j+1}^{r_j} x_k$;
\item
	the $j$-th job swapping increases the total completion time by at least $\sum_{k=\ell_j+1}^{r_j} x_k$;
	and the increment equals $\sum_{k=\ell_j+1}^{r_j} x_k$ if and only if $\ell_j > r_{j-1}$.
\end{itemize}
\end{lemma}
\begin{proof}
Note that $0 \le r_1 < r_2 < \ldots < r_m \le n$, $0 \le \ell_1 < \ell_2 < \ldots < \ell_m \le n$, and all these $\ell_j$'s and $r_j$'s are distinct from each other.
Since $\ell_j < r_j$, we assume without loss of generality that $r_{j'-1} < \ell_j < r_{j'}$ for some $j' \le j$, that is,
the $(j-j')$ jobs $J_{n+1 + r_{j'}}, J_{n+1 + r_{j'+1}}, \ldots, J_{n+1 + r_{j-1}}$ have been moved to be
in between $J_{\ell_j}$ and the center job $J_{2n+2}$ in the schedule $\pi^{j-1}$.

The $j$-th job swapping between the two jobs $J_{\ell_j}$ and $J_{n+1 + r_j}$ clearly decreases the total machine deterioration before the center job $J_{2n+2}$ by
$\delta_{\ell_j} - \delta_{r_j} = 2 \sum_{k=\ell_j+1}^{r_j} x_k$.

To estimate the total completion time,
we decompose the $j$-th job swapping between the two jobs $J_{\ell_j}$ and $J_{n+1 + r_j}$ as a sequence of $(r_j - \ell_j)$ ``{\em regular}'' job swappings,
between the two jobs $J_k$ and $J_{n+1 + k+1}$ for $k = r_j - 1, r_j - 2, \ldots, \ell_j + 1, \ell_j$.
We remark that the order of these regular job swappings is important, which guarantees that at the time of such a swapping,
the job $J_k$ is before the center job $J_{2n+2}$ and the job $J_{n+1 + k+1}$ is after the center job $J_{2n+2}$
(see the last paragraph before Section 3.1 for possible re-indexing the jobs).
For each such regular job swapping between the two jobs $J_k$ and $J_{n+1 + k+1}$, we can apply (almost, see below) Lemma~\ref{lemma03} to conclude that
it increases the total completion time by {\em at least} $x_{k+1}$.

From the proof of Lemma~\ref{lemma03},
the increment in the total completion time equals $x_{k+1}$ if and only if there are exactly $(n-k+1)$ jobs in between $J_{n+1 + k+1}$ and $J_n$ (inclusive),
that is, the $(j-j')$ jobs $J_{n+1 + r_{j'}}, J_{n+1 + r_{j'+1}}, \ldots, J_{n+1 + r_{j-1}}$ should not be moved
in between $J_k$ and the center job $J_{2n+2}$ in the schedule $\pi^{j-1}$.
Therefore, the $j$-th job swapping increases the total completion time by {\em at least} $\sum_{k=\ell_j+1}^{r_j} x_k$;
and the increment equals $\sum_{k=\ell_j+1}^{r_j} x_k$ if and only if $\ell_j > r_{j-1}$ ({\it i.e.}, $j' = j$).
This proves the lemma.
\end{proof}

\begin{lemma}
\label{lemma09}
For each $1 \le j \le m$, if the schedule $\pi^{j-1}$ satisfies the first two properties in Lemma~\ref{lemma02} and $\ell_j > r_j$, then
\begin{itemize}
\parskip=0pt
\item
	the schedule $\pi^j$ satisfies the first two properties in Lemma~\ref{lemma02};
\item
	the $j$-th job swapping increases the total machine deterioration before the center job $J_{2n+2}$ by
	$\delta_{r_j} - \delta_{\ell_j} = 2 \sum_{k=r_j+1}^{\ell_j} x_k$;
\item
	the $j$-th job swapping increases the total completion time by at least $\sum_{k=r_j+1}^{\ell_j} x_k$.
\end{itemize}
\end{lemma}
\begin{proof}
We prove first an analog to Lemma~\ref{lemma03} on a {\em regular} job swapping between the two jobs $J_{i_\ell + 1}$ and $J_{n+1 + i_\ell}$,
which can be viewed as an inverse operation of the regular job swapping between the two jobs $J_{i_\ell}$ and $J_{n+1 + i_\ell + 1}$.

For ease of presentation, let $C_i$ denote the completion time of the job $J_i$ in the schedule after the regular job swapping, and
let $C'_i$ denote the completion time of the job $J_i$ in the schedule before the regular job swapping.
Comparing to the schedule before the swapping,
\begin{itemize}
\parskip=0pt
\item
	the completion time of jobs preceding $J_{n+1 + i_\ell}$ is unchanged;
\item
	$C_{n+1 + i_\ell} - C'_{i_\ell + 1} = p_{i_\ell} - p_{i_\ell + 1} = - x_{i_\ell + 1}$;
\item
	the completion time of each job in between $J_{i_\ell + 2}$ and $J_n$ (inclusive, $n - i_\ell - 1$ of them) decreases by $x_{i_\ell + 1}$;
\item
	the duration of the first $\MA$ increases by $2 x_{i_\ell + 1}$;
\item
	the completion time of each job in between $J_{2n+2}$ and $J_{n+1+i_\ell + 1}$ (inclusive, $n - i_\ell + 1$ of them) increases by $x_{i_\ell + 1}$;
\item
	$C_{i_\ell + 1} - C'_{n+1+i_\ell} = x_{i_\ell + 1} + (\delta_{i_\ell + 1} + p_{i_\ell + 1}) - (\delta_{i_\ell} + p_{i_\ell}) = 0$;
\item
	consequently, the completion time of jobs succeeding $J_{i_\ell + 1}$ is unchanged.
\end{itemize}
The total completion time of jobs in the schedule after this regular job swapping increases by at least $x_{i_\ell + 1}$.
Note that the increment equals $x_{i_\ell + 1}$ if and only if
there are exactly $(n - i_\ell + 1)$ jobs in between $J_{2n+2}$ and $J_{n+1+i_\ell + 1}$ (inclusive),
that is, the $(j-1)$ jobs $J_{\ell_1}, J_{\ell_2}, \ldots, J_{\ell_{j-1}}$ should not be moved
in between the center job $J_{2n+2}$ and $J_{n+1 + i_\ell}$ in the schedule $\pi^{j-1}$.

Using the above analog of Lemma~\ref{lemma03}, the rest of the proof of the lemma is similar to the proof of Lemma~\ref{lemma08}
by decomposing the $j$-th job swapping between the two jobs $J_{\ell_j}$ and $J_{n+1 + r_j}$ as a sequence of $(\ell_j - r_j)$ ``{\em regular}'' job swappings,
between the two jobs $J_{k+1}$ and $J_{n+1 + k}$ for $k = \ell_j - 1, \ell_j - 2, \ldots, r_j + 1, r_j$.
\end{proof}

\begin{theorem}
\label{thm10}
If there is a feasible schedule $\pi$ to the instance $I$ with the total completion time no more than $Q = Q_0 + B$,
then there is a subset $X_1 \subset X$ of sum exactly $B$.
\end{theorem}
\begin{proof}
We start with a feasible schedule $\pi$, which has the first two properties stated in Lemma~\ref{lemma02} and
for which the total completion time is no more than $Q = Q_0 + B$.
Excluding the job $J_{2n+2}$, using the first $n+1$ jobs and the last $n+1$ job in $\pi$,
we determine the two subsets of indices $R = \{r_1, r_2, \ldots, r_m\}$ and $L = \{\ell_1, \ell_2, \ldots, \ell_m\}$,
and define the corresponding $m$ job swappings.
We then repeatedly apply the job swapping to convert the initial infeasible schedule $\pi^0$ into $\pi$.

In Case 1, the total machine deterioration of the first $(n+1)$ jobs in $\pi$ is
\[
\sum_{i=0}^n \delta_i - 2 \sum_{\ell_j < r_j} \sum_{k=\ell_j+1}^{r_j} x_k + 2 \sum_{\ell_j > r_j} \sum_{k=r_j+1}^{\ell_j} x_k 
	= \ML_0 - \delta,
\]
implying that
\begin{equation}
\label{eq6}
\sum_{\ell_j < r_j} \sum_{k=\ell_j+1}^{r_j} x_k - \sum_{\ell_j > r_j} \sum_{k=r_j+1}^{\ell_j} x_k = B + \frac 12 \delta,
\end{equation}
where $\delta \ge 0$ is the remaining machine maintenance level before the first $\MA$.

On the other hand, the total completion time of jobs in the schedule $\pi$ is at least
\[
Q_0 + \sum_{\ell_j < r_j} \sum_{k=\ell_j+1}^{r_j} x_k + \sum_{\ell_j > r_j} \sum_{k=r_j+1}^{\ell_j} x_k
	= Q_0 + B + \frac 12 \delta + 2 \sum_{\ell_j > r_j} \sum_{k=r_j+1}^{\ell_j} x_k.
\]
It follows that 1) $\delta = 0$; 2) there is no pair of swapping jobs $J_{\ell_j}$ and $J_{n+1 + r_j}$ such that $\ell_j > r_j$;
and 3) $\ell_1 < r_1 < \ell_2 < r_2 < \ldots < \ell_m < r_m$ (from the third item of Lemma~\ref{lemma08}).
Therefore, from Eq.~(\ref{eq6}), for the subset $X_1 = \cup_{j=1}^m \{x_{\ell_j+1}, x_{\ell_j+2}, \ldots, x_{r_j}\}$,
$\sum_{x \in X_1} x = B$.
That is, the instance $X$ of the {\sc Partition} problem is a yes-instance.

In Case 2, after all the $m$ job swappings, the first $\MA$ immediately precedes $J_{2n+2}$ and has its duration $- \delta$ since $\delta_{2n+2} = 0$,
where $\delta \ge 0$ is the remaining machine maintenance level before the first $\MA$.
$J_{2n+2}$ and its immediate succeeding $\MA$ and the following job need to be swapped to obtain the schedule $\pi$;
the thus moved $\MA$ is merged to the first $\MA$, resulting in a positive duration.
The total machine deterioration of the first $(n+1)$ jobs in $\pi$ (before the first $\MA$) is
\[
\sum_{i=0}^n \delta_i - 2 \sum_{\ell_j < r_j} \sum_{k=\ell_j+1}^{r_j} x_k + 2 \sum_{\ell_j > r_j} \sum_{k=r_j+1}^{\ell_j} x_k 
	= \ML_0 - \delta,
\]
implying that Eq.~(\ref{eq6}) still holds in this case.

On the other hand, the total completion time of jobs in the schedule $\pi$ is at least
\[
Q_0 + \sum_{\ell_j < r_j} \sum_{k=\ell_j+1}^{r_j} x_k + \sum_{\ell_j > r_j} \sum_{k=r_j+1}^{\ell_j} x_k + (\delta_{\pi_{n+2}} + p_{\pi_{n+2}} - p_{2n+2})
	\ge Q_0 + B + \frac 12 \delta + 2 \sum_{\ell_j > r_j} \sum_{k=r_j+1}^{\ell_j} x_k.
\]
The similarly as in Case 1, it follows that 1) $\delta = 0$; 2) there is no pair of swapping jobs $J_{\ell_j}$ and $J_{n+1 + r_j}$ such that $\ell_j > r_j$;
and 3) $\ell_1 < r_1 < \ell_2 < r_2 < \ldots < \ell_m < r_m$.
Therefore, from Eq.~(\ref{eq6}), for the subset $X_1 = \cup_{j=1}^m \{x_{\ell_j+1}, x_{\ell_j+2}, \ldots, x_{r_j}\}$,
$\sum_{x \in X_1} x = B$.
That is, the instance $X$ of the {\sc Partition} problem is a yes-instance.

In Case 3, after all the $m$ job swappings, the first $\MA$ immediately precedes $J_{2n+2}$ and has its duration $- \delta$ since $\delta_{2n+2} = 0$,
where $\delta \le 0$ is the remaining machine maintenance level before the first $\MA$.
Therefore, $J_{\pi_{n+1}}$ and the first $\MA$ need to be swapped to obtain the schedule $\pi$.
The total machine deterioration of the first $(n+1)$ jobs in $\pi$ (before the first $\MA$) is
\[
\sum_{i=0}^n \delta_i - 2 \sum_{\ell_j < r_j} \sum_{k=\ell_j+1}^{r_j} x_k + 2 \sum_{\ell_j > r_j} \sum_{k=r_j+1}^{\ell_j} x_k 
	= \ML_0 - \delta,
\]
implying that Eq.~(\ref{eq6}) still holds in this case, except that here $\delta \le 0$.

On the other hand, the total completion time of jobs in the schedule $\pi$ is at least
\[
Q_0 + \sum_{\ell_j < r_j} \sum_{k=\ell_j+1}^{r_j} x_k + \sum_{\ell_j > r_j} \sum_{k=r_j+1}^{\ell_j} x_k + (- \delta)
	\ge Q_0 + B - \frac 12 \delta + 2 \sum_{\ell_j > r_j} \sum_{k=r_j+1}^{\ell_j} x_k.
\]
The similarly as in Case 1, except that here $\delta \le 0$,
it follows that 1) $\delta = 0$; 2) there is no pair of swapping jobs $J_{\ell_j}$ and $J_{n+1 + r_j}$ such that $\ell_j > r_j$;
and 3) $\ell_1 < r_1 < \ell_2 < r_2 < \ldots < \ell_m < r_m$.
Therefore, from Eq.~(\ref{eq6}), for the subset $X_1 = \cup_{j=1}^m \{x_{\ell_j+1}, x_{\ell_j+2}, \ldots, x_{r_j}\}$,
$\sum_{x \in X_1} x = B$.
That is, the instance $X$ of the {\sc Partition} problem is a yes-instance.
\end{proof}

The following theorem follows immediately from Theorems~\ref{thm04} and \ref{thm10}.

\begin{theorem}
\label{thm11}
The general problem $(1 | p\MA | \sum_j C_j)$ is NP-hard.
\end{theorem}

\section{A $2$-approximation algorithm for $(1 | p\MA | \sum_j C_j)$}\label{sec3}
%===============================================================================
Recall that in the problem $(1 | p\MA | \sum_j C_j)$,
we are given a set of jobs ${\cal J} = \{J_i, i = 1, 2, \ldots, n\}$,
where each job $J_i = (p_i, \delta_i)$ is specified by its non-preemptive {\em processing time} $p_i$ and {\em machine deterioration} $\delta_i$. 
The machine deterioration $\delta_i$ quantifies the decrement in the machine maintenance level after processing the job $J_i$.
The machine has an initial machine maintenance level $\ML_0$, $0 \le \ML_0 \le \ML^{\max}$, where $\ML^{\max}$ is the maximum maintenance level.
The goal is to schedule the jobs and necessary $\MA$s of any duration such that all jobs can be processed without machine breakdown,
and that the total completion time of jobs is minimized.

In this section, we present a $2$-approximation algorithm, denoted as ${\cal A}_1$, for the problem.
Furthermore, the algorithm ${\cal A}_1$ produces a feasible schedule $\pi$ satisfying the first two properties stated in Lemma~\ref{lemma02},
suggesting that if the third property is violated then a local job swapping can further decrease the total completion time.

In the algorithm ${\cal A}_1$, the first step is to sort the jobs in SSF order (and thus we assume without loss of generality that)
$p_1 + \delta_1 \le p_2 + \delta_2 \le \ldots \le p_n + \delta_n$.
In the second step, the separation job is determined to be $J_k$, where $k$ is the maximum index such that $\sum_{i=1}^{k-1} \delta_i \le \ML_0$.
In the last step, the jobs preceding the separation job $J_k$ are re-sorted in the SPT order,
denoted by $(J_{i_1}, J_{i_2}, \ldots, J_{i_{k-1}})$, and the jobs succeeding the separation job are $(J_{k+1}, J_{k+2}, \ldots, J_n)$.
That is, the solution schedule is
\[
\pi = (J_{i_1}, J_{i_2}, \ldots, J_{i_{k-1}}; \MA_1, J_k; \MA_2, J_{k+1}, \MA_3, J_{k+2}, \ldots, \MA_{n-k+1}, J_n),
\]
where $\MA_1 = \sum_{j=1}^k \delta_j - \ML_0$ and $\MA_i = \delta_{k-1+i}$ for $i = 2, 3, \ldots, n-k+1$.

Let $\pi^*$ denote an optimal schedule satisfying all properties stated in Lemma~\ref{lemma02}, and its separation job is $J_{\pi^*_{k^*}}$:
\[
\pi^* = (J_{\pi^*_1}, J_{\pi^*_2}, \ldots, J_{\pi^*_{k^*-1}}; \MA^*_1, J_{\pi^*_{k^*}};
	\MA^*_2, J_{\pi^*_{k^*+1}}, \MA^*_3, J_{\pi^*_{k^*+2}}, \ldots, \MA^*_{n-k^*+1}, J_{\pi^*_n}).
\]
Let $C_i$ ($C^*_i$, respectively) denote the completion time of the job $J_{\pi_i}$ ($J_{\pi^*_i}$, respectively) in the schedule $\pi$ ($\pi^*$, respectively);
the makespans of $\pi$ and $\pi^*$ are $C_{\max}$ and $C^*_{\max}$, respectively, and (recall that $\ML_0 < \sum_{i=1}^n \delta_i$)
\begin{equation}
\label{eq7}
C_{\max} = C^*_{\max} = \sum_{i=1}^n (p_i + \delta_i) - \ML_0.
\end{equation}

\begin{lemma}
\label{lemma12}
For every $i \ge k$ we have
\[
\sum_{j=i}^n (p_j + \delta_j) \ge \sum_{j=i}^n (p_{\pi^*_j} + \delta_{\pi^*_j}).
\]
\end{lemma}
\begin{proof}
Since $p_1 + \delta_1 \le p_2 + \delta_2 \le \ldots \le p_n + \delta_n$,
$\sum_{j=i}^n (p_j + \delta_j)$ is the maximum sum of processing times and machine deterioration, over all possible subsets of $(n-i+1)$ jobs.
The lemma thus holds.
\end{proof}

\begin{theorem}
\label{thm13}
The algorithm ${\cal A}_1$ is an $O(n\log n)$-time $2$-approximation algorithm for the problem $(1 | p\MA | \sum_j C_j)$.
\end{theorem}
\begin{proof}
We compare the two schedules $\pi$ obtained by the algorithm ${\cal A}_1$ and $\pi^*$ an optimal schedule satisfying the properties stated in Lemma~\ref{lemma02}.
Using Eq.~(\ref{eq7}) and Lemma~\ref{lemma12}, it is clear that $C_i \le C^*_i$ for each $i = n, n-1, \ldots, \max\{k, k^*\}$.

Suppose $k < k^*$, then for each $i$ such that $k \le i < k^*$, we have
\begin{eqnarray*}
C_i &= 	&C_n - \sum_{j=i+1}^n (p_j + \delta_j)\\
	&\le&C^*_n - \sum_{j=i+1}^n (p_{\pi^*_j} + \delta_{\pi^*_j})\\
	&=  &\sum_{j=1}^i (p_{\pi^*_j} + \delta_{\pi^*_j}) - \ML_0\\
	&=  &\sum_{j=1}^i p_{\pi^*_j} - \left(\ML_0 - \sum_{j=1}^i \delta_{\pi^*_j}\right)\\
	&\le&\sum_{j=1}^i p_{\pi^*_j} = C^*_i. 
\end{eqnarray*}
Therefore, we have $C_i \le C^*_i$ for each $i = n, n-1, \ldots, k$.
It follows that
\begin{equation}
\label{eq8}
\sum_{i=k}^n C_i \le \sum_{i=k}^n C^*_i \le \mbox{OPT}.
\end{equation}

On the other hand, by the SPT order, the algorithm ${\cal A}_1$ achieves the minimum total completion time of jobs of $\{J_1, J_2, \ldots, J_{k-1}\}$.
One clearly sees that in the optimal schedule $\pi^*$, the sub-total completion time of $\{J_1, J_2, \ldots, J_{k-1}\}$ is upper-bounded by $\mbox{OPT}$.
Therefore,
\begin{equation}
\label{eq9}
\sum_{i=1}^{k-1} C_i \le \mbox{OPT}.
\end{equation}
Merging Eqs.~(\ref{eq8}) and (\ref{eq9}), we conclude that the total completion time of schedule $\pi$ is
\[
\sum_{i=1}^{k-1} C_i + \sum_{i=k}^n C_i \le 2 \cdot \mbox{OPT}.
\]
This proves the performance ratio of $2$ (which can also be shown tight on a trivial $2$-job instance
$I = \{J_1 = (1, \lambda), J_2 = (\lambda - 1, 1), \ML_0 = \ML^{\max} = \lambda\}$, with a sufficiently large $\lambda$).
The running time of the algorithm ${\cal A}_1$ is dominated by two times of sorting, each taking $O(n\log n)$ time. 
\end{proof}

\section{Concluding remarks}
%===============================================================================
We investigated the single machine scheduling with job-dependent machine deterioration, recently introduced by Bock {\it et al.}~\cite{BBH12},
with the objective to minimize the total completion time of jobs.
In the partial maintenance case, we proved the NP-hardness for the general problem, thus addressing the open problem left in the previous work.
From the approximation perspective, we designed a $2$-approximation, for which the ratio $2$ is tight on a trivial two-job instance.

The $2$-approximation algorithm is simple, but it is the first such work.
Our major contribution is the non-trivial NP-hardness proof, which might appear surprising at the first glance
since one has so much freedom in choosing the starting time and the duration of the maintenance activities.
It would be interesting to further study the (in-)approximability for the problem. %, for which we observed some contradicting phenomena to the existence of a PTAS.
It would also be interesting to study the problem in the full maintenance case, which was shown NP-hard, from the approximation algorithm perspective.
Approximating the problem in the full maintenance case seems more challenging, where we need to deal with multiple bin-packing sub-problems,
while the inter-relationship among them is much complex.

\section*{Acknowledgments}
%===============================================================================
W.L. was supported by the China Scholarship Council (Grant No. 201408330402).
Y.X. and G.L. were supported by NSERC.
W.T. was supported in part by funds from the Office of the Vice President for Research \& Economic Development and Georgia Southern University.

%===============================================================================
%\bibliography{../BiBTeX/scheduling,../BiBTeX/general}

\end{document}